\documentclass[lnicst]{svmultln}
\usepackage{makeidx}  % allows for indexgeneration
\usepackage{mathtools}
\usepackage{amsmath,amsfonts,amssymb}
\usepackage[dvips]{graphicx}

\begin{document}
\mainmatter              % start of the contribution
\title{Optimal Power Allocation for OFDM-Based Wire-Tap Channels with Arbitrarily Distributed Inputs}
\titlerunning{Optimal Power Allocation.}  % abbreviated title (for running head)
\author{Haohao Qin$^\star$, Yin Sun\thanks{Haohao Qin and Yin Sun contribute equally to this work.}, Xiang Chen, Ming Zhao, Jing Wang}

\authorrunning{Haohao Qin, Yin Sun et al.}   % abbreviated author list (for running head)
%Craig Chambers, Kim B. Bruce, Elisa Bertino}
%
\institute{State Key Laboratory on Microwave
and Digital Communications\\
Tsinghua National Laboratory for Information Science and
Technology\\
Department of Electronic Engineering, Tsinghua University, Beijing
100084, P.~R.~China\\
\email{{haohaoqin07, sunyin02, chenxiang98,
zhao.ming29}@gmail.com}\\
\email{wangj@mail.tsinghua.edu.cn} }

\maketitle              % typeset the title of the contribution
% \index{Ekeland, Ivar} % entries for the author index
% \index{Temam, Roger}  % of the whole volume
% \index{Dean, Jeffrey}

\begin{abstract}        % give a summary of your paper
%This paper investigates optimal power allocation for OFDM-based
%secure communications with arbitrary distributed channel inputs.
In this paper, we investigate power allocation that maximizes the
secrecy rate of orthogonal frequency division multiplexing (OFDM)
systems under arbitrarily distributed inputs. Considering commonly
assumed Gaussian inputs are unrealistic,
%In practical systems, discrete inputs are used rather than
%unrealistic Gaussian inputs. And for maximizing
%secrecy rate, power allocation strategies under Gaussian inputs yield quite poor performance when they are used under discrete inputs. %What's more, power
%%allocation strategies for maximizing secrecy rate under those
%%distributions are quite different. While Gaussian distributed
%%input signal achieves higher secrecy rate, it is unrealistic for its
%%infinite peak-to-average ratio.
%In view of this,
we focus on secrecy systems with more practical discrete distributed
inputs, such as PSK, QAM, etc. While the secrecy rate achieved by
Gaussian distributed inputs is concave with respect to the transmit
power, we have found and rigorously proved that the secrecy rate is
non-concave under any discrete inputs. Hence, traditional convex
optimization methods are not applicable any more. To address this
non-concave power allocation problem, we propose an efficient
algorithm. Its gap from optimality vanishes asymptotically at the
rate of $O(1/\sqrt{N})$, and its complexity grows in the order of
$O(N)$, where $N$ is the number of sub-carriers. Numerical results
are provided to illustrate the efficacy of the proposed algorithm.

\keywords {OFDM wire-tap channel, arbitrarily distributed inputs,
duality theory, non-convex problem, optimal power allocation}
\end{abstract}

\section{Introduction}
% no \PARstart
In recent years, many privacy sensitive wireless services, such as
pushmail, mobile wallet, Microblogging, etc, have become more and
more popular. But, due to the broadcast nature of wireless channels,
security problems and challenges are also accompanying with the
growing up of those privacy services.
% While it is convenient to access to these
%services through mobile phone, this also leads to more concerns of
%secrecy due to the easy wiretap of the subscribers' transmission
%signals in broadcast wireless channel.
The security of wireless communications is commonly supported by
cryptographic techniques employed at upper layer. However, this
traditional method faces several challenges, such as the emergence
of new cracking algorithms and increasing computational capability
of eavesdroppers.
%The eavesdropper can extract the key as long as it gets enough
%message \cite{Shannon1949}.
Recently, physical layer security, a method that can supplement
upper layer security,  has received considerable attentions
\cite{Liang2009}.

Physical layer security was firstly studied from an
information-theoretic perspective in \cite{Wyner1975}, where the
concept of ``wire-tap channel'' was introduced to illustrate the
channel with three terminals, transmitter, legitimate receiver and
eavesdroppers,
%Wyner
%considered a wire-tap channel with a transmitter, legitimate
%destination, and degraded eavesdropper.The
and secrecy rate was defined as the maximum achievable data rate
from the transmitter to its legitimate receiver while keeping the
eavesdropper completely ignorant of the secret massage.
% and is
%simply the difference between the mutual information of
%transmitter-destination channel and that of transmitter-eavesdropper
%channel in \cite{Wyner1975}.
Later, the research in this field was extended to %general broadcasting channel
%\cite{Csiszar1978}, which brought in deeper specific studies on
various scenes, such as Gaussian wire-tap channel
\cite{Csiszar1978}-\cite{Cheong1978}, multiple input multiple output
(MIMO) channel \cite{Oggier2007}-\cite{Ekrem2009}, orthogonal
frequency division multiplexing (OFDM) channel
\cite{Li06}-\cite{Renna2010}, etc.

Recently, OFDM-based secure communications have obtained much
attention for its capability of countermining the dispersive of
wideband wireless channels and enhance secrecy rate
\cite{Li06}-\cite{Renna2010}. Optimal power allocation of OFDM-based
wire-tap channels is investigated in \cite{Li06}-\cite{Jorwieck08}
under Gaussian inputs. In practical system, due to Gaussian inputs'
infinite peak-to-average ratio, finite discrete
constellations\footnote{The words ``distribution'' and
``constellation'' are used alternatively throughout the paper.},
such as PSK, QAM (see Fig.\ref{fig:SystemModel} (a)), are used
instead. In this paper, we investigate optimal power allocation for
OFDM-based wire-tap channels under arbitrarily distributed channel
inputs.

We have found and rigorously proved that the secrecy rate is
non-concave under any discrete constellations, while the secrecy
rate achieved by Gaussian distributed channel inputs was found to be
concave with respect to the transmit power
\cite{Li06}-\cite{Jorwieck08}. Therefore, the optimal power
allocation strategy for OFDM-based wire-tap channels with Gaussian
inputs \cite{Li06}-\cite{Jorwieck08} is not applicable any more to
the considered problem. To address this non-concave power allocation
problem, we propose an efficient power allocation algorithm. Its gap
from optimality vanishes asymptotically at the rate of
$O(1/\sqrt{N})$, and its complexity grows in the order of $O(N)$,
where $N$ is the number of sub-carriers. Numerical results are
provided to illustrate the efficiency of the proposed algorithm.

%\begin{figure}[t]
%\centering
%\includegraphics[scale=0.4]{DiscreteInput}
%\caption{(a). QPSK inputs.  (b). 16QAM inputs.}
%\label{fig:DiscreteInput}
%\end{figure}

The remainder of this paper is organized as follows. Section
\ref{sec:system model and secrecy rate} presents the system model
and power allocation problem is formulated. Optimal power allocation
for arbitrarily distributed inputs is given in Section
\ref{solution}. Numerical results and conclusions are provided in
Section \ref{sec:simulation} and Section \ref{sec:conclu},
respectively.

\section{System Model and Problem Formulation}\label{sec:system model and secrecy rate}
%\subsection{System Model}\label{sysmod}

\begin{figure}[t]
\centering
\includegraphics[scale=0.5]{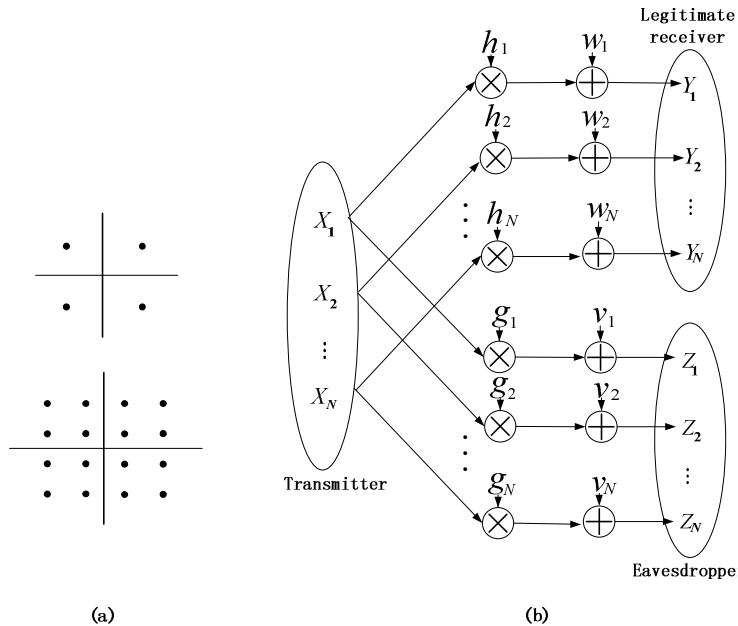}
\caption{(a). Finite discrete constellations: QPSK and 16QAM, (b).
OFDM-based wire-tap channel} \label{fig:SystemModel}
\end{figure}

Consider an OFDM-based wire-tap channel with a transmitter, a
legitimate receiver and an eavesdropper, where the eavesdropper
intends to extract the confidential message transmitted from
transmitter to legitimate receiver (see Fig. \ref{fig:SystemModel}
(b)). There are $N$ sub-carriers and the transmitter's signal in
each sub-carrier follows an arbitrary but predetermined
distribution, which can be either continuous constellations, such as
Gaussian distribution, or finite discrete constellations, including
PSK, QAM, etc (see Fig. \ref{fig:SystemModel} (a)).

The transmitted signal over the $i$th sub-carrier is denoted as
$x_i$, described as
\begin{equation}
x_i=\sqrt{p_i}s_i, i=1,\dots,N,
\end{equation}
where $p_i$ is the power ratio between transmission signal $x_i$ and
the noise, and $s_i$ represents the normalized channel inputs with
predetermined distribution. Then power constraint can be readily
shown to be
\begin{equation}\label{equ:power constraint}
\frac{1}{N}\sum_{i=1}^Np_i\leq P,
\end{equation}
where $P$ is total available transmit power.

The received signals at the legitimate receiver and eavesdropper are
given by
\begin{eqnarray}\label{Y}
&&y_i=h_i\sqrt{p_i}s_i+w_i,~~~~~i=1,\cdots,N,\\\label{Z}
&&z_i=g_i\sqrt{p_i}s_i+v_i,~~~~~~i=1,\cdots,N,
\end{eqnarray}
respectively, where the $w_i$ and $v_i$ are zero-mean complex
Gaussian noises with unit variance; $h_i$ and $g_i$ are the complex
channel coefficients of $i$th sub-carrier. According to the
information theoretical studies of \cite{Li06}, the secrecy rate
from transmitter to its legitimate receiver is
\begin{eqnarray}\label{Cp_S}
&&\sum_{i=1}^{N}[I(s_i;h_i\sqrt{p_i}s_i+w_i)-I(s_i;g_i\sqrt{p_i}s_i+v_i)]^+,
\end{eqnarray}
where $[x]^+ \triangleq \max\{x,0\}$, and $I(x;y)$ denotes the
mutual information between random variables $x$ and $y$. The
expression in (\ref{Cp_S}) is quite illuminating: the secrecy rate
of each sub-channel is non-negetive; if it is positive, it is
exactly the data rate difference of the legitimate and eavesdropping
channels; the total secrecy rate is simply the sum secrecy rate of
all the $N$ sub-carriers.

For fixed constellations of $\{s_i\}_{i=1}^N$, we need to optimize
the power allocation to obtain the maximal secrecy rate. The
optimization problem can be cast as follows, %\footnote{Unless
%otherwise stated, throughout the paper, the mutual information is
%expressed in nats per second per hertz (nats/s/Hz) and the
%logarithms are in natural base.}
\begin{equation}\label{problem1}
\begin{array}{ll}
R^* = \underset{\textbf{p}}{\max}~&R_s(\textbf{p})\triangleq \frac{1}{N} \sum\limits_{i=1}^{N}[I(s_i;h_i\sqrt{p_i}s_i+w_i)-I(s_i;g_i\sqrt{p_i}s_i+v_i)]^+\\[0.2cm]
~~~~~~~~~\textrm{s.t.}~~ &\frac{1}{N} \sum_{i=1}^N{p_i}\leq{P},\\[0.2cm]
&\textbf{p}\geq 0
%%& \sum_{i=1}^n\mathbb{E}[|S_i|^2]={1},\\[0.2cm]
%& Y_i=h_i\sqrt{p_i}S_i+W_i,~~ i=1,\cdots,N,\\[0.2cm]
%& Z_i=g_i\sqrt{p_i}S_i+V_i,~~ i=1,\cdots,N,\\[0.2cm]
%& p_i\geq 0,~~~~~~~~~~~~~~~~~ i=1,\cdots,N,\\[0.2cm]
\end{array}
\end{equation}
where $\textbf{p}\in\mathcal{R}^N$ is transmit power vector for $N$
sub-carriers, i.e., $\textbf{p}=\{p_1,p_2,...,p_N\}$, and $R^*$
denotes the optimal value. For the facility of the following
analysis, let $R_{s,i}(p_i)$ denote the ingredients of
$R_s(\textbf{p})$, i.e.,
\begin{equation}\label{equ:gradients of R}
R_{s,i}(p_i)\triangleq
[I(s_i;h_i\sqrt{p_i}s_i+w_i)-I(s_i;g_i\sqrt{p_i}s_i+v_i)]^+.
\end{equation}

\section{Optimal Power Allocation under Arbitrarily Distributed Channel Inputs}\label{solution}
\begin{figure}[t]
\centering
\includegraphics[scale=0.55]{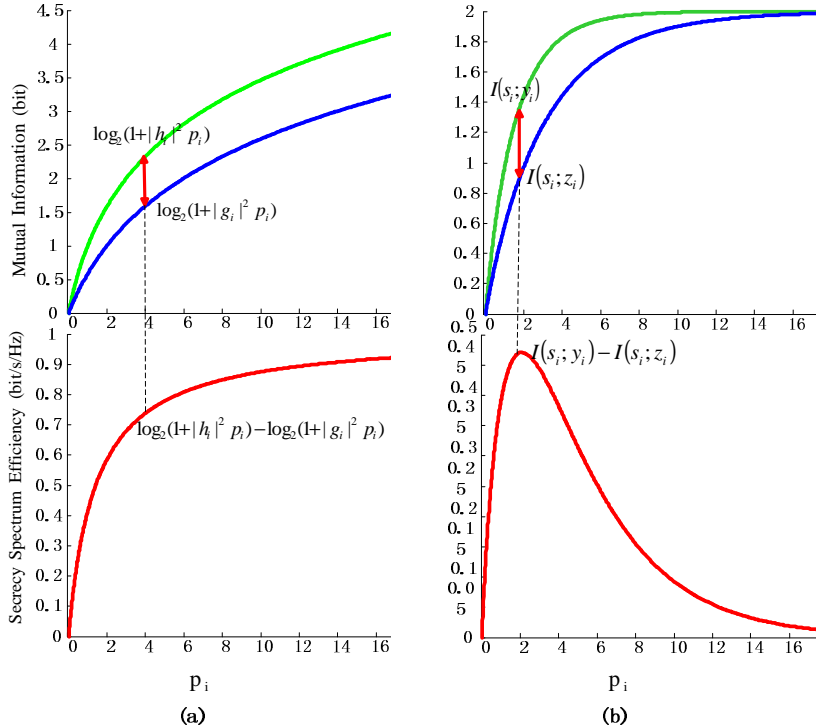}
\caption{(a). Secrecy rate achieved by Gaussian distributed inputs.
(b). Secrecy rate achieved by discrete inputs (eg. QPSK). Here
$h_i\geq g_i$.} \label{fig:I_input}
\end{figure}

%\begin{figure}[t]
%\centering
%\includegraphics[scale=0.4]{sdis}
%\caption{Discrete channel input} \label{fig:special inputs}
%\end{figure}
\subsection{Non-concavity of the Secrecy Rate $R_s(\textbf{p})$}\label{Sec:secrP}
If $s_i$ follows Gaussian distribution, the secrecy rate
$R_s(\textbf{p})$ in (\ref{problem1}) has explicit expression
\cite{Cheong1978}, i.e.,
\begin{equation}\label{c_gaussian}
R_s^G(\textbf{p})=\frac{1}{N} \sum_{i=1}^N[\;\log_2(1+{|h_i|^2p_i})-\log_2(1+{|g_i|^2p_i})\;]^+.\\
\end{equation}
It can be checked that $R^G_s(\textbf{p})$ is a concave function of
$\textbf{p}$. Hence problem (\ref{problem1}) is a convex
optimization problem. The ingredients of $R^G_s(\textbf{p})$ are
illustrated in the left part of Fig. \ref{fig:I_input}. One can
observe that $\log_2(1+{|h_i|^2p_i})$, $\log_2(1+{|g_i|^2p_i})$ and
$\log_2(1+{|h_i|^2p_i})-\log_2(1+{|g_i|^2p_i})$ are all concave,
provided that $|h_i|^2>|g_i|^2$. In \cite{Li06}-\cite{Jorwieck08},
utilizing the convexity structure of problem (\ref{problem1}),
authors obtained the optimal power allocation shown as follows,
\begin{eqnarray}\label{equ:gaussResult}
p_i^*=\begin{cases}
\frac{-(|h_i|^2+|g_i|^2)+\sqrt{(|h_i|^2+|g_i|^2)^2-4|h_i|^2|g_i|^2\frac{u+|g_i|^2-|h_i|^2}{u}}}{2|h_i|^2|g_i|^2}\\[0.1cm]
\hfill{,\textrm{if}~~|h_i|^2-|g_i|^2>u}\\[0.2cm]
 0 ~~~~~~~~~~~~~~~~~~~~~~~~~~~~~~~~~~ ,\textrm{others},
\end{cases}
\end{eqnarray}
where the Lagrange multiplier $u$ is chosen to meet the power
constraint,
\begin{equation}
\frac{1}{N} \sum_{i=1}^N{p_i}={P}.
\end{equation}

\begin{figure}[t]
\centering
\includegraphics[scale=0.55]{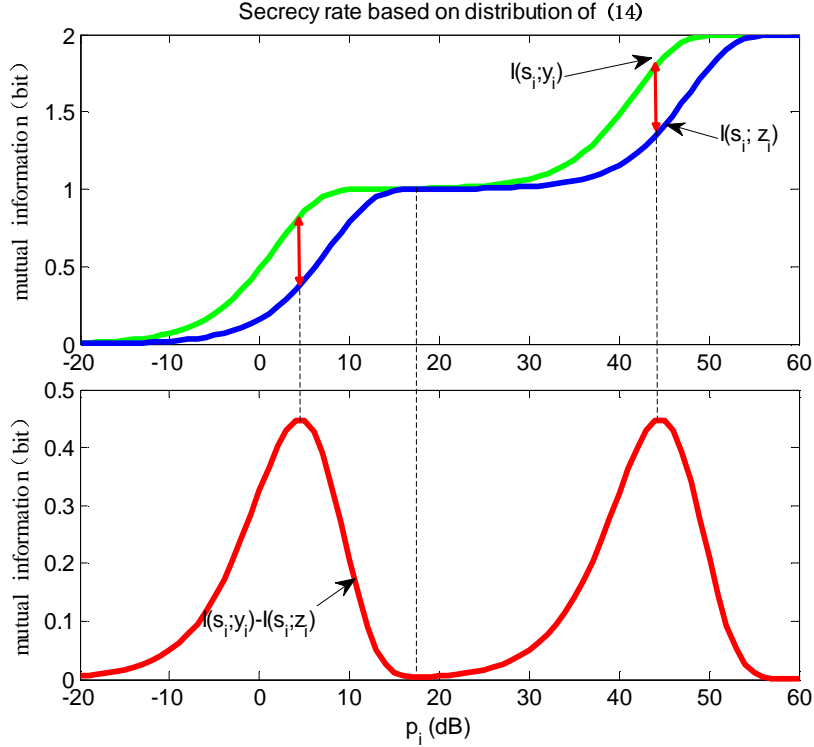}
\caption{Secrecy rate achieved by distribution of (\ref{dis:mess}).}
\label{fig:I_special distributed}
\end{figure}

%\begin{figure}[t]
%\centering
%\includegraphics[scale=0.55]{I_special_distribution}
%caption{}\label{fig:I_special distributed}
%\end{figure}

One may expect that the concavity of $R_s(\textbf{p})$ still holds
under general input distributions. Unfortunately, our investigation
shows that this is not true, which is formally presented in the
following proposition.

\begin{proposition}
The secrecy rate function $R_s(\textbf{p})$ under any finite
discrete constellations is non-concave with respect to \textbf{p}.
\end{proposition}

\begin{proof}
When $p_i=0$, one can derive $I(s_i;y_i)=I(s_i;z_i)=0$; when
$p_i=+\infty$, we have $I(s_i;y_i)=I(s_i;z_i)=H(s_i)$, where $H(x)$
is entropy of $x$. Therefore, $R_{s,i}({0})=R_{s,i}(+\infty)=0$.

According to \cite{Guo05},
\begin{equation}\label{equ:derivative}
\frac{\partial{I(s;\sqrt{p}s+n)}}{\partial{p}}=\texttt{MMSE}(p).
\end{equation}
$\texttt{MMSE}(p)$ is defined as
\begin{eqnarray}\label{equ:MMSE}
&& \texttt{MMSE}(p)\triangleq
\mathbb{E}[|s-\mathbb{E}(s|\sqrt{p}{s}+n)|^2],
\end{eqnarray}
where $\mathbb{E}[x]$ is the expectation of random variable $x$;
$\mathbb{E}[x|y]$ is the conditional expectation of $x$ for given
$y$. By (\ref{equ:derivative}), the derivative of $R_{s,i}({p_i})$
at $p_i=0$ is given by\footnote{Only the sub-carriers that satisfy
$|h_i|^2>|g_i|^2$ are considered, as $R_{s,i}(p_i)\equiv 0$ for
those sub-carriers with $|h_i|^2\leq|g_i|^2$, which do not affect
the concavity of $R_s(\textbf{p})$.}
\begin{equation}
R_{s,i}'({p_i})|_{{p_i}={0}}={\left[|h_i|^2-|g_i|^2\right]^+}>0,%\frac{\partial{R_{s,i}({p_i})}}{\partial{{p_i}}}
\end{equation}
which indicates that there must exist a $\hat{p_i}>0$ that
$R_{s,i}(\hat{p_i})>0$. %It is known that $R_{s,i}(\hat{p_i})$ is
%derivatable with respective to $p_i$ \cite{Guo05},
According to the Lagrange's mean value theorem \cite{Jeffreys1988},
it must have a point $\tilde{p_i}\in [\hat{p_i},+\infty]$ with
negative slop $R_{s,i}'(\tilde{p_i})<0$.

Assume $R_{s,i}(p_i)$ is concave, then the inequality
$R_{s,i}(p_i)\leq R_{s,i}(\tilde{p_i}) +
R_{s,i}'(\tilde{p_i})(p_i-\tilde{p_i})$ holds \cite{Boyd2003}, which
indicates $R_{s,i}(+\infty)=-\infty $. This contradicts
$R_{s,i}(+\infty)=0$. Therefore, the concavity assumption is not
true, and Proposition 1 holds.
\end{proof}

Two evidentiary examples are provided to illustrate Proposition 1.

The first example is QPSK. The curves of $I(s_i;y_i)$, $I(s_i;z_i)$
and $R_{s,i}(p_i)$ versus $p_i$ are shown in right part of Fig.
\ref{fig:I_input}, and they are in accordance with the statements in
the proof of Proposition 1.

The second example considers a 4 points PAM constellation with
non-uniform spacing. The probability mass function of $s_i$ is given
by
\begin{equation}\label{dis:mess}
P_{s_i}\sim
\begin{bmatrix}
&-51L &~~~ -50L ~~& ~~~50L ~~~~& 51L \\
&0.25 & 0.25 &0.25 &0.25
\end{bmatrix},
\end{equation}
where $L$ is a normalization parameter to maintain unit variance.
Figure \ref{fig:I_special distributed} shows the secrecy rate
$R_{s,i}(p_i)$ for this case. It is interesting that the
$R_{s,i}(p_i)$ has two peaks. Hence, it is definitely non-concave.
We note that the mutual information $I(s_i;y_i)$ and $I(s_i;z_i)$
are concave with respect to $p_i$ in linear scale \cite{Guo2010}.

\subsection{Optimal Power Allocation Solution of Problem (\ref{problem1})}
Although problem (6) is non-convex, there are still some efficient
algorithms to solve it and obtain near-optimal solutions. One of
them is the Lagrangian duality method \cite{Boyd2003}. Some recent
studies \cite{Luo2009}-\cite{Yu2006} showed that, for certain
non-convex structures, asymptotic optimal performance can be
achieved by this method.

The Lagrangian of problem (\ref{problem1}) is given by
\begin{equation}
L(\textbf{p},u)={\frac{1}{N}\sum_{i=1}^N[I(s_i;h_i\sqrt{p_i}s_i+w_i)-I(s_i;g_i\sqrt{p_i}s_i+v_i)]^++u\left(P-\frac{1}{N}\sum_{i=1}^Np_i\right)},
\end{equation}
where $u$ is Lagrangian dual variable. The corresponding dual
function can then be written as
\begin{equation}\label{pro:inner prob.}
\begin{array}{ll}
g(u)&\triangleq \underset{\textbf{p}\geq 0}\max{~~L(\textbf{p},u)}.
\end{array}
\end{equation}
%Where
%\begin{equation}\label{pro:ithc}
%g_i(u)=\underset{p_i\geq
%0}\max{[I(S_i;Y_i)-I(S_i;Z_i)-up_i]},i=1,2,...,N.
%\end{equation}
Hence the dual problem formulation of problem (\ref{problem1}) can
be readily expressed as
\begin{equation}\label{pro:dual}
\begin{array}{ll}
D^*=&\underset{u\geq 0}\min~~g(u),
\end{array}
\end{equation}
where $D^*$ denotes the optimal dual value. % When the
%objective function of primal problem (\ref{problem1}) is concave in
%$\{p_i\}_{i=1}^n$, the optimal value $R^*$ is equal to the optimal
%value of (\ref{pro:dual}), $D^*$, which leads to that one can solve
%problem (\ref{pro:dual}) to get the optimal solution for primal
%problem (\ref{problem1}).
Since the objective function of primal problem (\ref{problem1}) is
non-concave, there is a positive gap between $R^*$ and $D^*$, i.e.,
$D^*-R^*>0$. However, according to the recent studies of Luo and
Zhang \cite{Luo2009}, \cite{Luo2008}, asymptotic strong duality
holds for problem (\ref{problem1}), i.e. the duality gap $D^*-R^*$
goes to zero as $N\rightarrow\infty$, as is expressed in the
following proposition:

\begin{proposition}
If the channel coefficients $g_i$ and $h_i$ are Lipschitz continuous
and bounded in the sense
\begin{equation} \label{equ:hi_lipschitz}
\big| |h_i|-|h_j| \big| \leq L_h \frac{|i-j|}{N}, \forall~
i,j\in\{1,2,...,N\}
\end{equation}
\begin{equation}\label{equ:gi_lipschitz}
\big||g_i|-|g_j|\big| \leq L_g \frac{|i-j|}{N}, \forall~
i,j\in\{1,2,...,N\}
\end{equation}
where $L_h,L_g>0$ is the Lipschitz constant. Then we have
\begin{equation}
0\leq D^*-R^*\leq O\left(\frac{1}{\sqrt{N}}\right).
\end{equation}
\end{proposition}

\begin{proof} According to (\ref{equ:derivative}) and
(\ref{equ:MMSE}), we have \cite{Guo05}
\begin{equation}\label{equ:derivative_bound}
0 \leq
\frac{\partial{I(s;\sqrt{p}s+n)}}{\partial{p}}=\texttt{MMSE}(p)\leq
\mathbb{E}[|s|^2]=1,
\end{equation}
which implies that
\begin{equation}\label{equ:derivative_R_p}
\frac{\partial{R_{s,i}(p)}}{\partial{p}}=|h_i|^2\texttt{MMSE}(|h_i|^2p)-|g_i|^2\texttt{MMSE}(|g_i|^2p)\leq
|h_i|^2,
\end{equation}
\begin{equation}\label{equ:derivative_R_hi}
\frac{\partial{R_{s,i}(p)}}{\partial{|h_i|}}=2p|h_i|\texttt{MMSE}(|h_i|^2p)\leq
2p|h_i|,
\end{equation}
\begin{equation}\label{equ:derivative_R_gi}
\frac{\partial{R_{s,i}(p)}}{\partial{|g_i|}}=2p|g_i|\texttt{MMSE}(|g_i|^2p)\leq
2p|g_i|,
\end{equation}
where $R_{s,i}(p)$ is defined in the expression of
(\ref{equ:gradients of R}).

According to equation (\ref{equ:derivative_R_p}) and the Lagrange's
mean value theorem \cite{Jeffreys1988}, we have
\begin{equation}\label{equ:R_lipschis_p}
|R_{s,i}(p)-R_{s,i}(p')|\leq |h_i|^2||p-p'||_\infty.
\end{equation}
On the other hand, combining (\ref{equ:derivative_R_hi}),
(\ref{equ:derivative_R_gi}), (\ref{equ:derivative_bound}),
(\ref{equ:hi_lipschitz}) and (\ref{equ:gi_lipschitz}), using chain
rule and the Lagrange's mean value theorem \cite{Jeffreys1988}, we
have
\begin{equation}\label{equ:R_lipschis_i}
|R_{s,i}(p)-R_{s,j}(p)|\leq \left(2p|h_i|L_h+2p|g_i|L_g\right)
\frac{|i-j|}{N}.
\end{equation}
As $|h_i|$, $|g_i|$ and $p$ are bounded, according to
(\ref{equ:R_lipschis_p}) and (\ref{equ:R_lipschis_i}), there must
exist an $L$ such that
\begin{equation}
|R_{s,i}(p)-R_{s,j}(p')|\leq
L\left(\frac{|i-j|}{N}+||p-p'||_\infty\right)
\end{equation}
i.e., the secrecy rate $R_{s,i}(p)$ is Lipschitz continuous. Hence
according to Theorem 2 of \cite{Luo2009}, the duality gap between
$D^*$ and $R^*$ vanishes asymptotically in the order of
$1/\sqrt{N}$, which is expressed as
\begin{equation}
0\leq D^*-R^* \leq O(\frac{1}{\sqrt{N}}),
\end{equation}
and Proposition 2 holds.
\end{proof}

The procedures to solve (\ref{pro:inner prob.}) and (\ref{pro:dual})
are provided in the following.

\begin{table}
\caption{} \label{tab1} \centering
\begin{tabular}{l}
\hline Algorithm : Lagrangian dual optimization method\\
\hline Initialize $u$\\
\emph{repeat}\\
~~~~for i=1 to N \\
~~~~~~~ find $p_i=\arg\underset{p_i}\max{\left[[I(s_i;y_i)-I(s_i;z_i)]^+-up_i\right]+uP}$. \\
~~~~end\\
~~~~update $u$ using bisection method.\\[0.1cm]
\emph{until} $u$ converges.\\
\hline
\end{tabular}
\end{table}

For each fixed $u$, problem (\ref{pro:inner prob.}) can be decoupled
into $N$ independent sub-carrier problems
%\begin{equation}
%g(u)=\sum_{i=1}^N\underset{p_i\geq
%0}\max{\left[[I(s_i;y_i)-I(s_i;z_i)]^+-up_i\right]+uP}.
%\end{equation}
\begin{align}\label{pro:subp}
g(u)& = \underset{p_i\geq 0}\max{~~L(p_i,u)},\nonumber\\
&=\sum_{i=1}^N\underset{p_i\geq
0}\max{\left[[I(s_i;y_i)-I(s_i;z_i)]^+-up_i\right]+uP}.
\end{align}
While the sub-carrier problem in (\ref{pro:subp}) is still
non-convex, it has only one variable $p_i$ and can be solved by
simple one dimension line search.
%Note that each sub-carrier optimization problem is non-convex
%problem. So, the solution to sub-carrier optimization problem
%simplifies into an exhaustive search.
 As the dual function $g(u)$ is
convex in $u$, %\footnote{This is because $L(p_i,u)$ is linear function
%in $u$ for each fixed $p_i$, and $g(u)$ takes the maximum of those
%linear functions. So, it is convex.}
its subgradient $g'(u)=P-\frac{1}{N}\sum_{i=1}^Np_i^*$, where
$p_i^*$ is optimal solution for problem (\ref{pro:inner prob.}) with
fixed $u$, is an increasing function in $u$. Hence bisection method
can be used to solve dual problem (\ref{pro:dual}), so that either
$u=0$, $P\geq\frac{1}{N}\sum_{i=1}^Np_i^*$ or $u>0$,
$P=\frac{1}{N}\sum_{i=1}^Np_i^*$ is satisfied.
% with derivative
%$\frac{\partial{g(u)}}{\partial
%u}|_{u=u_0}=P-\frac{1}{N}\sum_{i=1}^Np_i^*$, where $p_i^*$ is
%optimal solution for problem (\ref{pro:inner prob.}) with $u=u_0$.% $\frac{\partial{g(u)}}{\partial
%u}|_{u=u_0}=P-\frac{1}{N}\sum_{i=1}^Np_i^*$,
Table.\ref{tab1} summarizes the algorithm.

The complexity of this algorithm is $N \frac{1}{e_p} \log_2
(\frac{1}{e_d})$, where $e_p$ is the accuracy of one dimension
exhausitive search to solve (\ref{pro:inner prob.}) and $e_d$ is the
accuracy of the bisection search to solve (\ref{pro:dual}). Since
its complexity is linear with respect to the number of sub-carriers
$N$, it is quite convenient for practical large values of $N$, such
as 64$\sim$4096. We note that the complexity of solve
(\ref{problem1}) exhaustively is $\frac{1}{e_p^N}$, which is
exponential in $N$ and thus unrealistic.

%Compared with solving primal problem (\ref{problem1}), the solution
%for dual problem (\ref{pro:dual}) has lower complexity. This is
%because of the following two reasons. First, decouple procedure of
%problem (\ref{pro:inner prob.}) leads to a complexity that is linear
%in $N$. Second, bisection search method solving convex problem
%(\ref{pro:inner prob.}) is linear complexity, which is independent
%on $N$. According to those two facts, the proposed Lagrangian dual
%optimization method has an $O(N)$ complexity.

%We note that applying the proposed algorithm to Gaussian case, we
%can obtain the optimal solution (\ref{equ:gaussResult}). Our method
%is suitable for both convex and non-convex problem.

\section{Numerical Results}\label{sec:simulation}
\begin{figure}[t]
\centering
\includegraphics[scale=0.6]{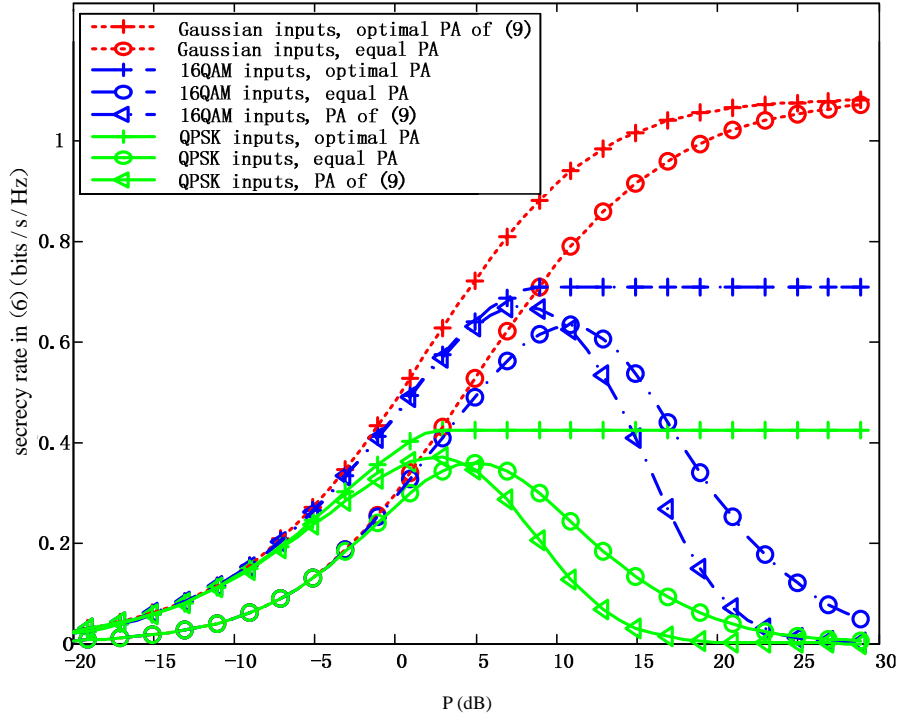}
\caption{The secrecy rate versus total power $P$.}
\label{Fig:comparison}
\end{figure}

%\begin{figure}[t]
%\begin{minipage}[t]{0.6\linewidth}
%\centering
%\includegraphics[scale=0.4]{bar_QPSK}
%\caption{Secrecy rate based on Gaussian distributed inputs and
%discrete distributed inputs.} \label{Fig:bar_QPSK}
%\end{minipage}
%\begin{minipage}[t]{0.4\linewidth}
%\centering
%\includegraphics[scale=0.4]{bar_Gau}
%\caption{Secrecy rate based on Gaussian distributed inputs and
%discrete distributed inputs.} \label{Fig:bar_Gau}
%\end{minipage}
%\end{figure}

In this section, we provide some simulation results to illustrate
the performance of our proposed power allocation algorithm and show
how different channel input distributions affect the secrecy rate
and power allocation results.

\begin{figure}[t]
\centering
\includegraphics[scale=0.43]{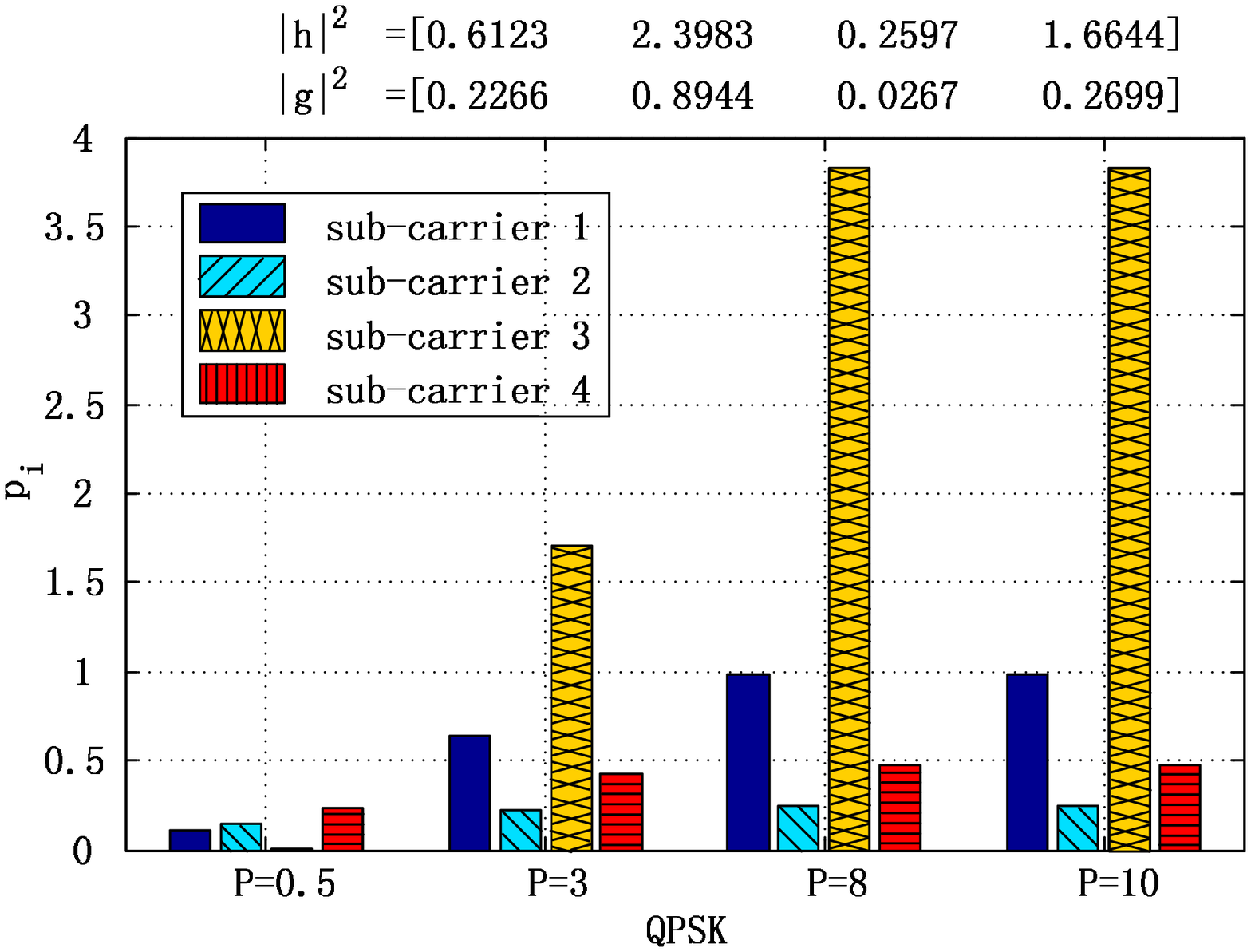}
\caption{Power allocation results versus $P$ under QPSK
inputs.}\label{Fig:bar_QPSK}
\end{figure}

\begin{figure}[t]
\centering
\includegraphics[scale=0.45]{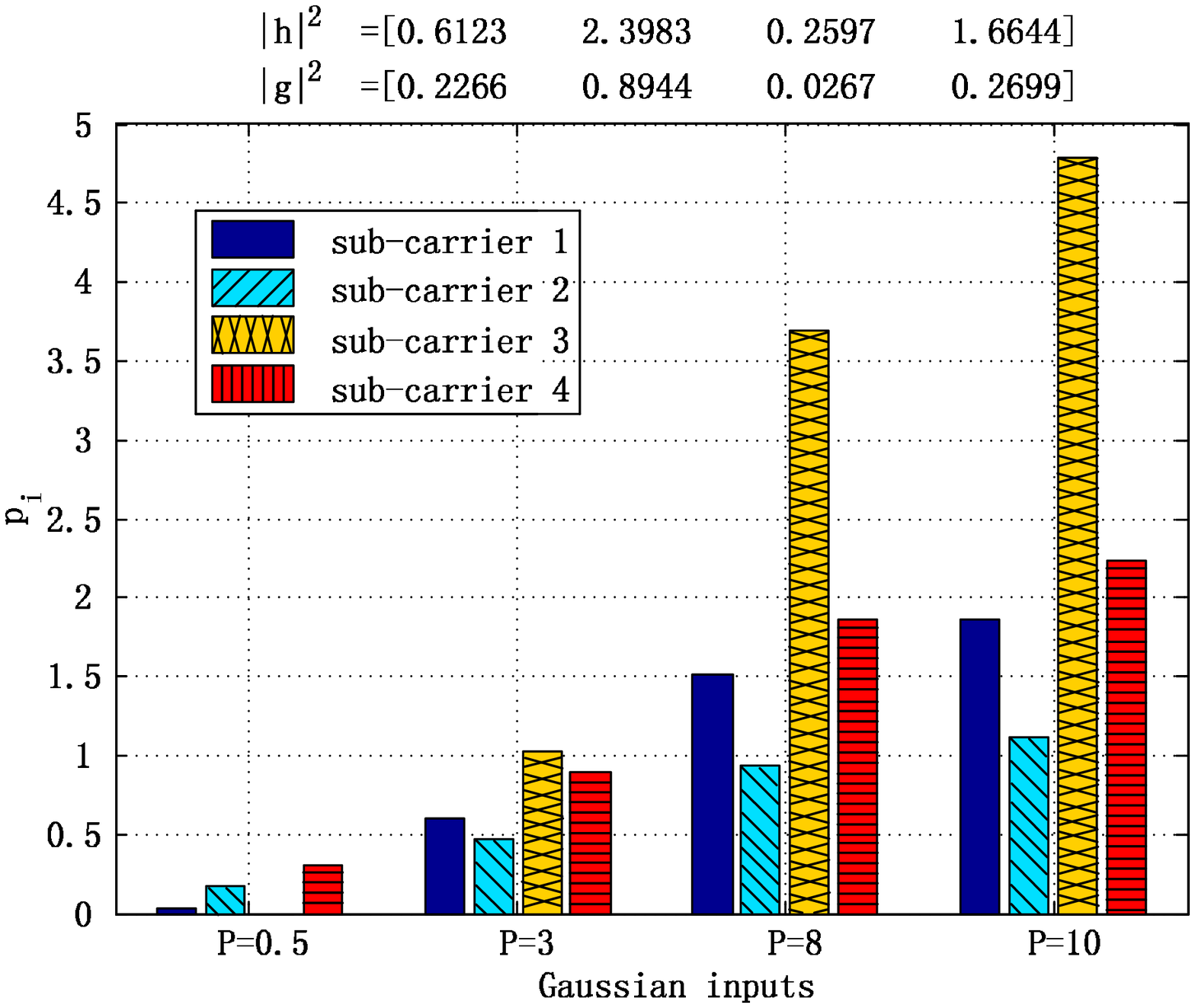}
\caption{Power allocation results versus $P$ under Gaussian
inputs.}\label{Fig:bar_Gau}
\end{figure}

We first consider an OFDM-based secure system with $N=128$
sub-carriers. The secrecy rate versus total power constraint under
different power allocation strategies and channel input
distributions are illustrated in Fig. \ref{Fig:comparison}. Two
reference strategies are considered to compare with our strategy:
the optimal strategy under Gaussian inputs, i.e.,
(\ref{equ:gaussResult}), which is denoted by ``PA of
(\ref{equ:gaussResult})'' in Fig. \ref{Fig:comparison}; the equal
power allocation strategy, which equally allocates total power among
the sub-carriers that satisfy $|h_i|^2>|g_i|^2$ and is denoted by
``equal PA''.

Seen from Fig. \ref{Fig:comparison}, higher secrecy rate can be
achieved for QPSK and 16QAM by our proposed optimal power allocation
strategy, especially when the power constraint $P$ is quite large.
However, equal power allocation and power allocation of
(\ref{equ:gaussResult}) can be quite bad under finite discrete
constellations, and the secrecy rate even drops to zero for large
value of $P$. Actually, when $P$ is large, secrecy rate under
Gaussian inputs can be approximated by
\begin{equation}
\begin{array}{ll}
R_s^G(\textbf{p})& =\frac{1}{N}
\sum_{i=1}^N[\;\log_2(1+{|h_i|^2p_i})-\log_2(1+{|g_i|^2p_i})\;]^+\\[0.25cm]
& \approx\frac{1}{N}
\sum_{i=1}^N[\;\log_2(\frac{|h_i|^2}{|g_i|^2})\;]^+,
\end{array}
\end{equation}
which is independent with power allocation $p_i$. So equal power
allocation works asymptotically optimal under Gaussian inputs when
$P$ is large, as shown in Fig. \ref{Fig:comparison}.

%For comparison, we apply waterfilling strategy, i.e.,
%(\ref{equ:gaussResult}), which is optimal for system with Gaussian
%inputs, equal power allocation strategy, which equally allocates
%total power among the subcarriers that satisfy $|h_i|^2>|g_i|^2$,
%and our proposed power allocation strategy in Table \ref{tab1} for
%QPSK, 16QAM and Gaussian inputs. The resulting values of secrecy
%rate versus maximum total power are respectively shown in
%Fig.\ref{Fig:comparison}, from which, we have following
%observations:

%1). Compared with other two power allocation strategies, our
%proposed asymptotic optimal power allocation method (benchmarks with
%'$\_$o') always achieves highest secrecy rate under same power
%constraint. This gain is very large when $P$ is high.
%
%2). When $P$ is high, for Gaussian inputs, equal power allocation
%strategy can perform as well as optimal power allocation strategy.
%In contrast, for QPSK/16QAM inputs, equal power allocation strategy
%shows a severe negative gain when $P$ is high, the reason of which
%is illustrated in the proof of Lemma 1.
%
%For Gaussian inputs, when $P$ is large, secrecy rate in
%(\ref{equ:gaussResult}) can be approximated by
%\begin{equation}
%\begin{array}{ll}
%R_s^G(\textbf{p})& =\frac{1}{N}
%\sum_{i=1}^N[\;\log_2(1+{|h_i|^2p_i})-\log_2(1+{|g_i|^2p_i})\;]^+\\[0.25cm]
%& \approx\frac{1}{N}
%\sum_{i=1}^N[\;\log_2(\frac{|h_i|^2}{|g_i|^2})\;]^+,
%\end{array}
%\end{equation}
%which is independent with power allocation $p_i$.

The power allocation solution of the proposed algorithm is shown in
Fig. \ref{Fig:bar_QPSK} and Fig. \ref{Fig:bar_Gau}, respectively,
under QPSK and Gaussian inputs with $N=4$. When the power constraint
$P$ is small, most transmit power is allocated to the stronger
sub-channels, the channels with larger $|h_i|^2-|g_i|^2$ (Channel 2
and Channel 4 in our simulation example). However, as $P$ grows, the
transmit power allocated to the weak sub-channels grows. Under QPSK
input signals, the transmit power allocated to every sub-channel
should stop increasing when $P$ is very large. But the transmit
power under Guassian input signals keeps increasing.

\section{Conclusion}\label{sec:conclu}
In this paper, we have obtained the optimal power allocation for
OFDM-based wire-tap channels with arbitrarily distributed inputs.
While the secrecy rate achieved by Gaussian distributed channel
inputs is concave with respect to the transmit power, we have found
and rigorously proved that the secrecy rate is non-concave under any
practical finite discrete constellations. A power allocation
algorithm has been proposed, which is asymptotic optimal as the
number of sub-carrier increases. Our numerical results show that
more transmit power may results in a huge loss in secrecy rate,
which is rarely seen in previous power allocation studies. This
indicates that optimal power allocation is quite essential in
practical studies of physical layer security.

%In this paper, we have obtained the optimal power allocation for
%OFDM-based wire-tap channels with arbitrary distributed inputs. We
%analysis the properties of secrecy rate with arbitrary distributed
%inputs, that is more transmitting power may result in secrecy rate
%loss with improper power allocation. We give optimal power
%allocation that maximize secrecy rate and apply to arbitrary
%distributed inputs. The algorithm to obtain optimal power allocation
%is derived. It is shown that optimal power allocation algorithm
%grantees secrecy rate does not decrease with more transmitting
%power, and provides more secrecy rate.

\subsubsection{Acknowledge}
The authors would like to thank Dr. Tsung-Hui Chang and Prof.
Shidong Zhou for their valuable suggestions in this paper.

This work is supported by National S\&T Major Project
(2009ZX03002-002), National Basic Research Program of China
(2007CB310608), National S\&T Pillar Program (2008BAH30B09),
National Natural Science Foundation of China (60832008) and PCSIRT,
Tsinghua Research Funding-No.2010THZ02-3. This work is also
sponsored by Datang Mobile Communications Equipment Co., Ltd.

%This work is supported by Tsinghua-Qualcomm Joint Research Program,
%National S\&T Major Project (2008ZX03O03-004), National Basic
%Research Program of China (2007CB310608), China¡¯s 863 Project
%(2009AA011501), National Natural Science Foundation of China
%(60832008) and PCSIRT, Tsinghua Research Funding-No.2010THZ02-3.

%\bibitem{fos:kes:2} Foster, I., Kesselman, C., Nick, J., Tuecke, S.: The Physiology of the
%Grid: an Open Grid Services Architecture for Distributed Systems
%Integration. Technical report, Global Grid Forum (2002)

%

\begin{thebibliography}{1}

%\bibitem{Shannon1949} C. E. Shannon, ``Communication Theory of Secrecy Systems,'' {\sl Bell. Syst. Tech. J.}, vol. 28, no. 4, pp. 656-715, Oct. 1949.

\bibitem{Liang2009}Y. Liang, H. V. Poor, and S. Shamai (Shitz), ``Information theoretic
security,'' {\sl Found. Trends Commun. Inf. Theory,} vol. 5, pp.
355-580, 2008.

\bibitem{Wyner1975} A. Wyner, ``The wire-tap channel,'' {\sl Bell Syst. Tech. J.}, vol. 54, no. 8, pp. 1355-1387, Jan. 1975.

\bibitem{Csiszar1978}I. Csiszar and J. Korner, ``Broadcast channels
with confidential messages,'' {\sl IEEE Trans. Inf. Theory}, vol.
24, no. 3, pp. 339-348, May 1978.

\bibitem{Cheong1978}S. L. Y. Cheong and M. Hellman, ``The Gaussian
wire-tap channel,'' {\sl IEEE Trans. Inf. Theory}, vol. 24, no. 4,
pp. 451-456, July 1978.


\bibitem{Oggier2007} F. Oggier, and B. Hassibi, ``The secrecy capacity of the MIMO wiretap channel,''
{\sl in Proc. 45th Annu. Allerton Conf. Communication, Control and Computing}, Monticello, IL, Sept. 2007, pp. 848-855.

\bibitem{Liu2009} T. Liu, and
S. Shammai(Shitz), ``A note on the secrecy capacity of the
multi-antenna wiretap channel,'' {\sl IEEE Trans. Inf. Theory}, vol.
55, no. 6, pp. 2547-2553, Jun. 2009.

\bibitem{Ekrem2009} E. Ekrem, and S. Ulukus, ``Gaussian MIMO multi-receiver wiretap channel,''
{\sl  Global Telecommunications Conf.,} Honolulu, HI, Nov. 2009.



\bibitem{Li06} Z. Li, R. Yates, and W. Trappe,
``Secrecy capacity of independent parallel channels,'' {\sl in Proc.
44th Annu. Allerton Conf.}, Allerton House, Illinois, pp. 841-848,
Jul. 2006.

\bibitem{Liang2008} Y. Liang, H. V. Poor and S. Shamai(Shitz), ``Secure communication over fading channels,'' {\sl IEEE Trans. Inf. Theory.}, vol. 54, no. 6, pp. 2470-2492, Jun. 2008.


\bibitem{Jorwieck08}
E. Jorswieck and A. Wolf, ``Resource allocation for the wire-tap
multi-carrier broadcast channel,'' {\sl in Proc. International
Workshop Multiple Access Communications Conf.}, St. Petersburg,
Russia, June 2008.

\bibitem{Renna2010}
F. Renna, N. Laurenti and H. V. Poor, ``Physical layer security for
OFDM systems,'' {\sl European Wireless Conf.},Vienna, Austria, Apr.
2011.



\bibitem{Luo2009} Z. Luo, and S. Zhang, ``Duality Gap
Estimation and Polynomial Time Approximation for Optimal Spectrum
Management,'' {\sl IEEE Trans. Signal Process.}, vol. 57, no. 7, pp.
2675-2689, July 2009.

\bibitem{Luo2008} Z. Luo and S. Zhang, ``Dynamic spectrum management: Complexity
and duality,''{\sl IEEE J. Sel.Top. Signal Process., Special Issue
on Signal Process. Netw. Dyn. Spectrum Access,} vol. 2, no. 1, pp.
57-73, Feb. 2008.

\bibitem{Yu2006} W. Yu and R. Lui, ``Dual methods for nonconvex spectrum optimization of multicarrier systems,'' {\sl IEEE Trans. Commun.}, vol. 54, pp. 1310-1322, July. 2006.


\bibitem{Guo05}
D. Guo, S. Shamai (Shitz), and S. Verdu, ``Mutual information and
minimum mean-square error in Gaussian channels,'' {\sl IEEE Trans.
Inf. Theory}, vol. 51, no. 4, pp. 1261-1283, Apr. 2005.

%\bibitem{Lozano2006}
%A.~Lozano, A.~M. Tulino, and S.~Verdu, ``Optimum power allocation
%for parallel Gaussian channels with arbitrary input distributions,''
%{\sl IEEE Trans. Inf. Theory}, vol. 52, no. 7, pp. 3033-3051, Jul.
%2006.
%
%\bibitem{Lozano2010}
%F. Perez-Cruz, M. R. D. Rodrigues, and S.~Verdu, ``MIMO Gaussian
%Channels With Arbitrary Inputs: Optimal Precoding and Power
%Allocation,'' {\sl IEEE Trans. Inf. Theory}, vol. 56, no. 3, pp.
%1070-1084, Mar. 2010.

\bibitem{Guo2010}
D. Guo, Y. Wu, S. Shamai, and S. Verdu, ``Estimation in Gaussian
noise: Properties of the minimum mean-square error,'' {\sl IEEE
Trans. Inf. Theory}, vol. 57, no. 4, pp. 2371-2385, April 2011.

\bibitem{Jeffreys1988} H. Jeffreys and B. S. Jeffreys, {\sl Methods of Mathematical Physics,} Cambridge, England: Cambridge University Press, 3rd ed., 1988.



\bibitem{Boyd2003}
L. Boyd, S. Vandenberghe, {\sl Convex Optimization,} Cambridge,
England: Cambridge University Press, 2004.

%
%%
%% ---- Bibliography ----
%%
%\begin{thebibliography}{5}
%
%\bibitem{smit:wat} Smith, T.F., Waterman, M.S.: Identification of Common Molecular
%Subsequences. J. Mol. Biol. 147, 195--197 (1981)
%
%\bibitem{mes} May, P., Ehrlich, H.C., Steinke, T.: ZIB Structure Prediction Pipeline:
%Composing a Complex Biological Workflow through Web Services. In: Nagel,
%W.E., Walter, W.V., Lehner, W. (eds.) Euro-Par 2006. LNCS, vol. 4128,
%pp. 1148--1158. Springer, Heidelberg (2006)
%
%\bibitem{fos:kes} Foster, I., Kesselman, C.: The Grid: Blueprint for a New Computing
%Infrastructure. Morgan Kaufmann, San Francisco (1999)
%
%\bibitem{cff} Czajkowski, K., Fitzgerald, S., Foster, I., Kesselman, C.: Grid
%Information Services for Distributed Resource Sharing. In: 10th IEEE
%International Symposium on High Performance Distributed Computing, pp.
%181--184. IEEE Press, New York (2001)
%
%\bibitem{fos:kes:2} Foster, I., Kesselman, C., Nick, J., Tuecke, S.: The Physiology of the
%Grid: an Open Grid Services Architecture for Distributed Systems
%Integration. Technical report, Global Grid Forum (2002)
%
%\bibitem{url} National Center for Biotechnology Information, http://www.ncbi.nlm.nih.gov

\end{thebibliography}
\end{document}